\documentclass[envcountsame]{llncs}

\usepackage{amssymb, amsmath,graphicx,graphics,enumerate, tkz-graph}
\usepackage[hypertexnames=false]{hyperref}
\usepackage[T1]{fontenc}
\usepackage[utf8]{inputenc}
\usepackage{lmodern}
\usepackage{microtype}
\usepackage[noadjust,nospace]{cite}

\DeclareMathOperator{\bull}{bull}

\DeclareMathOperator{\cw}{cw}

\newcommand{\ssi}{\subseteq_i}

\newcommand{\li}{\subseteq_{li}}

\newcounter{ctrclaim}[theorem]
\newcounter{ctrobs}[theorem]
\newcounter{ctrcase}[theorem]
\newcounter{ctrsubcase}[ctrcase]

\makeatletter
\@addtoreset{ctrclaim}{lemma} 
\@addtoreset{ctrcase}{lemma}
\@addtoreset{ctrobs}{lemma}
\makeatother
\newcounter{ctrfake}
\makeatletter
\@addtoreset{ctrclaim}{ctrfake} 
\@addtoreset{ctrcase}{ctrfake}
\@addtoreset{ctrobs}{ctrfake}
\makeatother

\newenvironment{enumeratei}{\begin{enumerate}[(i)]}{\end{enumerate}}
\newcommand\figurenames{Figs.}
\newcommand\faketheorem[1]{\refstepcounter{ctrfake}{\bf #1}}
\newcommand\displaycase[1]{{\bf #1}}

\newcommand{\clm}[1]{\medskip\phantomsection\refstepcounter{ctrclaim}\noindent\displaycase{Claim \thectrclaim.}{\em #1}\\}
\newcommand{\obs}[1]{\medskip\phantomsection\refstepcounter{ctrobs}\noindent\displaycase{Observation \thectrobs.}{\em #1}\\}

\title{Bounding the Clique-Width of ${H}$-free Split Graphs
\thanks{An extended abstract of this paper appeared in the proceedings of EuroComb 2015~\cite{BDHP15b}.
The research in this paper was supported by EPSRC (EP/K025090/1).
The third author is grateful for the generous support of the Graduate (International) Research Travel Award from Simon Fraser University and Dr. Pavol Hell's NSERC Discovery Grant.}
}
\author{Andreas Brandst{\"a}dt\inst{1} \and Konrad K. Dabrowski\inst{2} \and\\ Shenwei Huang\inst{3} \and Dani\"el Paulusma\inst{2}}
\institute{
Institute of Computer Science, Universität Rostock,\\
Albert-Einstein-Straße 22, 18059 Rostock, Germany\\
\texttt{ab@informatik.uni-rostock.de}
\and
School of Engineering and Computing Sciences, Durham University,\\
Science Laboratories, South Road,
Durham DH1 3LE, United Kingdom
\texttt{\{konrad.dabrowski,daniel.paulusma\}@durham.ac.uk}
\and
School of Computing Science, Simon Fraser University,\\
8888 University Drive, Burnaby B.C., V5A 1S6, Canada\\
\texttt{shenweih@sfu.ca}
}

\begin{document}
\maketitle
\setcounter{footnote}{0}

\begin{abstract}
A graph is $H$-free if it has no induced subgraph isomorphic to~$H$.
We continue a study into the boundedness of clique-width of subclasses of perfect graphs. 
We identify
five new classes of $H$-free split graphs whose clique-width is bounded.
Our main result, obtained by combining new and known results, provides a classification of all but 
two stubborn cases, that is, with two potential exceptions we determine {\em all} graphs~$H$ for which the class of $H$-free split graphs has bounded clique-width.
\end{abstract}

\section{Introduction}\label{sec:intro}

Clique-width is a well-studied graph parameter; see for example the surveys of
Gurski~\cite{Gu07} and Kami\'nski, Lozin and Milani\v{c}~\cite{KLM09}. A graph
class is said to be of {\it bounded} clique-width if there is a constant~$c$
such that the clique-width of every graph in the class is at most~$c$.  Much
research has been done identifying whether or not various classes have bounded
clique-width~\cite{BL02,BGMS14,BDHP15,BELL06,BKM06,BK05,BLM04b,BLM04,BM02,BM03,DGP14,DHP0,DLRR12,DP14,DP15,GR99b,LR04,LR06,LV08,MR99}.
For instance, the Information System on Graph Classes and their
Inclusions~\cite{isgci} maintains a record of graph classes for which this is
known. In a recent series of papers~\cite{BDHP15,DHP0,DP15} the
clique-width of graph classes characterized by two forbidden induced subgraphs
was investigated. In particular we refer to~\cite{DP15} for  details on how new
results can be combined with known results to give a classification for all
but~$13$ open cases (up to an equivalence relation).  Similar studies have been
performed for variants of clique-width, such as linear
clique-width~\cite{HMP12} and power-bounded clique-width~\cite{BGMS14}.
Moreover, the (un)boundedness of the clique-width of a graph class seems to be
related to the computational complexity of the {\sc Graph Isomorphism} problem,
which has in particular been investigated for graph classes defined by two
forbidden induced subgraphs~\cite{KS12,Sc15}.
Indeed, a common technique (see e.g.~\cite{KLM09}) for showing that a class of graphs has unbounded clique-width relies on showing that it contains simple path encodings of walls or of graphs in some other specific graph class known to have unbounded clique-width.
Furthermore, Grohe and Schweitzer~\cite{GS15} recently proved that {\sc Graph Isomorphism} is polynomial-time solvable on graphs of bounded clique-width.

In this paper we continue a study into the boundedness of clique-width of
subclasses of perfect graphs.
Clique-width is still a very difficult graph parameter to deal with.
For instance, deciding whether or not a graph has clique-width at most~$c$ for some
fixed constant~$c$ is only known to be polynomial-time solvable if $c\leq
3$~\cite{CHLRR12}, but is a long-standing open problem for $c\geq 4$.
Our long-term goal is to increase our understanding of clique-width.
To this end we aim to identify new classes of bounded clique-width.
In order to explain some previously known results, along with our new ones, we first give some terminology.

\medskip
\noindent
{\bf Terminology.} 
For two vertex-disjoint graphs~$G$ and~$H$, the {\it disjoint union} $(V(G)\cup V(H), E(G)\cup E(H))$ is denoted by~$G+\nobreak H$ and the disjoint union of~$r$ copies of~$G$ is denoted by~$rG$. The {\it complement} of a graph~$G$, denoted by~$\overline{G}$, has vertex set $V(\overline{G})=V(G)$ and an edge between two distinct vertices
if and only if these vertices are not adjacent in~$G$. 
For two graphs~$G$ and~$H$ we write $H\ssi G$ to indicate that~$H$ is an induced subgraph of~$G$.
The graphs $C_r,K_r,K_{1,r-1}$ and~$P_r$ denote the cycle, complete graph, star and path on~$r$ vertices, respectively.
The graph~$S_{h,i,j}$, for $1\leq h\leq i\leq j$, denotes the {\em subdivided claw}, that is
the tree that has only one vertex~$x$ of degree~$3$ and exactly three leaves, which are of distance~$h$,~$i$ and~$j$ from~$x$, respectively.
For a set of graphs $\{H_1,\ldots,H_p\}$, a graph~$G$ is {\em $(H_1,\ldots,H_p)$-free} if it has no induced subgraph isomorphic to a graph in $\{H_1,\ldots,H_p\}$.
The {\it bull} is the graph with vertices $a,b,c,d,e$ and edges $ab,bc,ca,ad,be$; the {\em dart} is the graph obtained from the bull by adding the edge~$bd$ (see also \figurename~\ref{fig:bull-dart}).
\begin{figure}
\centering
\begin{tabular}{cc}
\begin{minipage}{0.33\textwidth}
\begin{center}
\scalebox{0.75}{
{\begin{tikzpicture}[scale=1,rotate=90]
\GraphInit[vstyle=Simple]
\SetVertexSimple[MinSize=6pt]
\Vertex[x=0,y=0]{a}
\Vertex[a=30,d=1]{b}
\Vertex[a=30,d=2]{c}
\Vertex[a=-30,d=1]{d}
\Vertex[a=-30,d=2]{e}
\Edges(c,b,a,d,e)
\Edges(b,d)
\end{tikzpicture}}}
\end{center}
\end{minipage}
&
\begin{minipage}{0.33\textwidth}
\begin{center}
\scalebox{0.75}{
{\begin{tikzpicture}[scale=1,rotate=270]
\GraphInit[vstyle=Simple]
\SetVertexSimple[MinSize=6pt]
\Vertex[x=0,y=0]{a}
\Vertex[a=60,d=1]{b}
\Vertex[a=-60,d=1]{c}
\Vertex[a=0,d=1]{d}
\Vertex[a=0,d=-1]{y}
\Edges(y,a,b,d,c,a,d)
\end{tikzpicture}}}
\end{center}
\end{minipage}\\
\\
bull & dart
\end{tabular}
\caption{\label{fig:bull-dart} The bull and the dart.}
\end{figure}
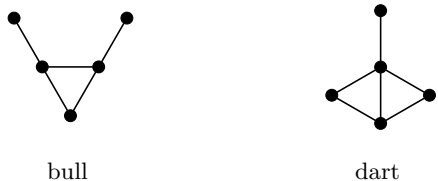

A graph~$G$ is {\em perfect} if, for every induced subgraph $H\ssi G$, the chromatic number of~$H$ equals its clique number.
By the Strong Perfect Graph Theorem~\cite{CRST06}, a graph~$G$ is perfect if and only if both~$G$ and~$\overline{G}$ are $(C_5,C_7,C_9,\ldots)$-free.
A graph~$G$ is {\em chordal} if it is $(C_4,C_5,\ldots)$-free and {\em weakly chordal} if both~$G$ and~$\overline{G}$ are $(C_5,C_6,\ldots)$-free. 
Every split graph is chordal, every chordal graph is weakly chordal and every weakly chordal graph is perfect.

\medskip
\noindent
{\bf  Known Results on Subclasses of Perfect Graphs.}
We start off with the following known theorem, which shows that the restriction of $H$-free graphs to $H$-free weakly chordal graphs does not yield any new graph classes of bounded clique-width, as both classifications are exactly the same.

\begin{theorem}[\cite{BDHP15,DP15}]\label{thm:weakly-chordal}
Let~$H$ be a graph. The class of
$H$-free (weakly chordal) graphs has bounded clique-width if and only if~$H$ is an induced subgraph of~$P_4$.
\end{theorem}

Motivated by Theorem~\ref{thm:weakly-chordal} we investigated classes of $H$-free chordal graphs in an attempt to identify {\it new} classes of bounded clique-width and 
as a (successful) means to find reductions to solve more cases in our classification for $(H_1,H_2)$-free graphs.
This classification for classes of $H$-free chordal graphs is almost complete except for two cases, which we call~$F_1$ and~$F_2$ (see \figurename~\ref{fig:bounded-split} for a definition).

\begin{theorem}[\cite{BDHP15}]\label{thm:chordal-classification}
Let~$H$ be a graph not in $\{F_1,F_2\}$. The class of $H$-free chordal graphs
has bounded clique-width if and only if
\begin{itemize}
\item [$\bullet$] $H=K_r$ for some $r\geq 1$;
\item [$\bullet$] $H\ssi \bull$;
\item [$\bullet$] $H\ssi P_1+P_4$;
\item [$\bullet$] $H\ssi \overline{P_1+P_4}$;
\item [$\bullet$] $H\ssi \overline{K_{1,3}+2P_1}$;
\item [$\bullet$] $H\ssi P_1+\overline{P_1+P_3}$;
\item [$\bullet$] $H\ssi P_1+\overline{2P_1+P_2}$ or
\item [$\bullet$] $H\ssi \overline{S_{1,1,2}}$.
\end{itemize}
\end{theorem}
In contrast to chordal graphs, the classification for bipartite graphs, another class of perfect graphs, is complete. This classification was used in the proof of 
Theorem~\ref{thm:chordal-classification} and it is similar to a characterization of Lozin and Volz~\cite{LV08} for a different variant of the notion of $H$-freeness in bipartite graphs (see~\cite{DP14} for an explanation of the difference 
between $H$-free bipartite graphs and the so-called strongly $H$-free bipartite graphs considered in~\cite{LV08}).

\begin{theorem}[\cite{DP14}]\label{thm:bipartite}
Let~$H$ be a graph. The class of $H$-free bipartite graphs has bounded
clique-width if and only~if
\begin{itemize}
\item [$\bullet$] $H=sP_1$ for some $s\geq 1$;
\item [$\bullet$]  $H\ssi K_{1,3}+3P_1$;
\item [$\bullet$]  $H\ssi K_{1,3}+P_2$;
\item [$\bullet$]  $H\ssi P_1+S_{1,1,3}$ or
\item [$\bullet$]  $H\ssi S_{1,2,3}$.
\end{itemize}
\end{theorem}

\noindent
{\bf Our Results.} We consider subclasses of split graphs.
A graph $G=(V,E)$ is a {\em split graph} if it has a {\em split partition}, that is, a partition of~$V$ into two (possibly empty) sets~$K$ and~$I$, where~$K$ is a clique and~$I$ is an independent set. The class of split graphs coincides with the class of 
$(2K_2,C_4,C_5)$-free graphs~\cite{FH77} and is known to have unbounded clique-width~\cite{MR99}.
As with the previous graph classes, we forbid one additional induced subgraph~$H$. We aim  to classify the boundedness of clique-width for $H$-free split graphs and to identify new graph classes of bounded clique-width along the way.
Theorem~\ref{thm:chordal-classification} also provides motivation, as it would be useful to know whether or not the clique-width of $H$-free split graphs is bounded when $H=F_1$ or $H=F_2$ (the two missing cases for chordal graphs;
recall that chordal graphs form a superclass of split graphs).
We give affirmative answers for both of these cases.
It should be noted that,
for any graph~$H$ the class of $H$-free split graphs has bounded clique-width if
and only if the class of $\overline{H}$-free split graphs has bounded
clique-width
(see also Lemma~\ref{lem:complement}).
As such our main result  shows that there are only two open cases (see also \figurenames~\ref{fig:bounded-split} and~\ref{fig:open-split}).

\begin{theorem}\label{thm:split-classification}
Let~$H$ be a graph such that neither~$H$ nor~$\overline{H}$ is in $\{F_4,F_5\}$.
The class of $H$-free split graphs has bounded clique-width if and only if
\\[-15pt]
\begin{itemize}
\item [$\bullet$] $H$ or~$\overline{H}$ is isomorphic to~$rP_1$ for some $r \geq 1$;
\item [$\bullet$] $H$ or $\overline{H} \ssi \bull+P_1$;
\item [$\bullet$] $H$ or $\overline{H} \ssi F_1$; 
\item [$\bullet$] $H$ or $\overline{H} \ssi F_2$; 
\item [$\bullet$] $H$ or $\overline{H} \ssi F_3$;
\item [$\bullet$] $H$ or $\overline{H} \ssi Q$ or
\item [$\bullet$] $H$ or $\overline{H} \ssi K_{1,3}+2P_1$.
\end{itemize}
\end{theorem}

\begin{figure}
\begin{center}
\begin{tabular}{ccc}
\begin{minipage}{0.25\textwidth}
\centering
\scalebox{0.75}{
{\begin{tikzpicture}[scale=1,rotate=135]
\GraphInit[vstyle=Simple]
\SetVertexSimple[MinSize=6pt]
\Vertices{circle}{a,b,c,d}
\end{tikzpicture}}}
\end{minipage}
&
\begin{minipage}{0.25\textwidth}
\centering
\scalebox{0.75}{
{\begin{tikzpicture}[scale=1,rotate=90]
\GraphInit[vstyle=Simple]
\SetVertexSimple[MinSize=6pt]
\Vertex[x=0,y=0]{a}
\Vertex[a=30,d=1]{b}
\Vertex[a=30,d=2]{c}
\Vertex[a=-30,d=1]{d}
\Vertex[a=-30,d=2]{e}
\Vertex[x=1.73205080757,y=0]{x}
\Edges(c,b,a,d,e)
\Edges(b,d)
\end{tikzpicture}}}
\end{minipage}
&
\begin{minipage}{0.25\textwidth}
\centering
\scalebox{0.75}{
{\begin{tikzpicture}[scale=1,rotate=135]
\GraphInit[vstyle=Simple]
\SetVertexSimple[MinSize=6pt]
\Vertices{circle}{a,b,c,d}
\Vertex[a=45,d=1.57313218497]{e}
\Vertex[a=225,d=1.57313218497]{f}
\Edges(a,b,c,d,a,c)
\Edges(b,d)
\Edges(a,e)
\Edges(f,d)
\end{tikzpicture}}}
\end{minipage}\\
\\
$rP_1$~for~$r=\nobreak 4$ & $\bull+\nobreak P_1$ & $F_1$\\
\\
\\
\end{tabular}
\begin{tabular}{cccc}
\begin{minipage}{0.24\textwidth}
\centering
\scalebox{0.75}{
{\begin{tikzpicture}[scale=1,rotate=135]
\GraphInit[vstyle=Simple]
\SetVertexSimple[MinSize=6pt]
\Vertices{circle}{a,b,c,d}
\Vertex[a=45,d=1.57313218497]{e}
\Vertex[a=225,d=1.57313218497]{f}
\Edges(a,b,c,d,a,c)
\Edges(b,d)
\Edges(a,e)
\Edges(c,f,d)
\end{tikzpicture}}}
\end{minipage}
&
\begin{minipage}{0.24\textwidth}
\centering
\scalebox{0.75}{
{\begin{tikzpicture}[xscale=-1,rotate=135]
\GraphInit[vstyle=Simple]
\SetVertexSimple[MinSize=6pt]
\Vertices{circle}{a,b,c,d}
\Vertex[x=1,y=2]{y}
\Vertex[x=2,y=1]{z}
\Edges(a,b,c,d,a,c)
\Edges(b,d)
\Edges(a,z,b,y)
\end{tikzpicture}}}
\end{minipage}
&
\begin{minipage}{0.24\textwidth}
\centering
\scalebox{0.75}{
{\begin{tikzpicture}[scale=1,rotate=30]
\GraphInit[vstyle=Simple]
\SetVertexSimple[MinSize=6pt]
\Vertex[a=0,d=0.57735026919]{a}
\Vertex[a=120,d=0.57735026919]{b}
\Vertex[a=240,d=0.57735026919]{c}
\Vertex[a=60,d=1.15470053838]{d}
\Vertex[a=180,d=1.15470053838]{e}
\Vertex[a=300,d=1.15470053838]{f}
\Edges(a,b,c,a,d,b,e)
\Edge(c)(f)
\end{tikzpicture}}}
\end{minipage}
&
\begin{minipage}{0.24\textwidth}
\centering
\scalebox{0.75}{
{\begin{tikzpicture}[scale=1]
\GraphInit[vstyle=Simple]
\SetVertexSimple[MinSize=6pt]
\Vertex[x=0,y=1.4142135623]{a}
\Vertex[x=-0.70710678118,y=0]{b}
\Vertex[x=0,y=0]{c}
\Vertex[x=0.70710678118,y=0]{d}
\Vertex[x=-0.70710678118,y=1.4142135623]{e}
\Vertex[x=0.70710678118,y=1.4142135623]{f}
\Edges(c,a,b)
\Edges(a,d)
\end{tikzpicture}}}
\end{minipage}
\\
& & \\
$F_2$ & $F_3$ & $Q$ & $K_{1,3}+\nobreak 2P_1$ 
\end{tabular}
\end{center}
\caption{The graphs~$H$ from Theorem~\ref{thm:split-classification} for which the classes of $H$-free split graphs and $\overline{H}$-free split graphs have bounded clique-width.}
\label{fig:bounded-split}
\end{figure}
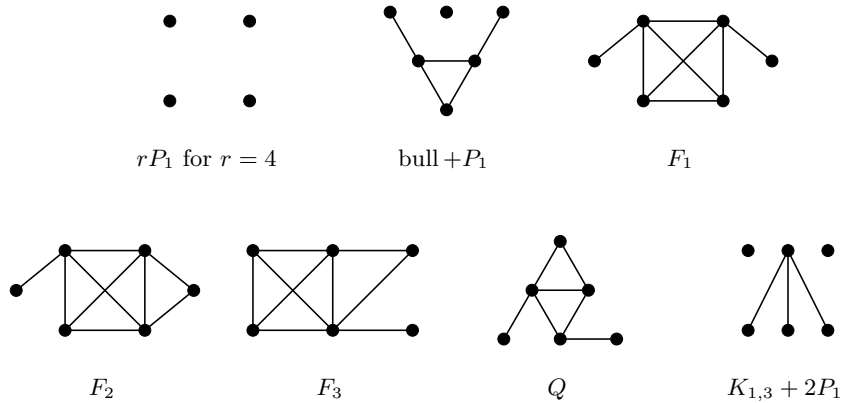
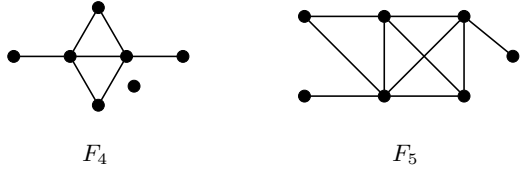
\begin{figure}
\begin{center}
\begin{tabular}{cc}
\begin{minipage}{0.33\textwidth}
\centering
\scalebox{0.75}{
{\begin{tikzpicture}[scale=1]
\GraphInit[vstyle=Simple]
\SetVertexSimple[MinSize=6pt]
\Vertex[x=0,y=0]{a}
\Vertex[a=60,d=1]{b}
\Vertex[a=-60,d=1]{c}
\Vertex[a=0,d=1]{d}
\Vertex[a=0,d=2]{x}
\Vertex[a=0,d=-1]{y}
\Vertex[a=-25,d=1.25]{z}
\Edges(y,a,b,d,c,a,d,x)
\end{tikzpicture}}}
\end{minipage}
&
\begin{minipage}{0.33\textwidth}
\centering
\scalebox{0.75}{
{\begin{tikzpicture}[scale=1,rotate=135]
\GraphInit[vstyle=Simple]
\SetVertexSimple[MinSize=6pt]
\Vertices{circle}{a,b,c,d}
\Vertex[a=225,d=1.57313218497]{f}
\Vertex[x=1,y=2]{y}
\Vertex[x=2,y=1]{z}
\Edges(a,b,c,d,a,c)
\Edges(b,d)
\Edges(a,z,b,y)
\Edges(f,d)
\end{tikzpicture}}}
\end{minipage}\\
& \\
$F_4$ & $F_5$
\end{tabular}
\end{center}
\caption{\label{fig:open-split} The (only) two graphs for which it is not known whether or not the classes of $H$-free split graphs and $\overline{H}$-free split graphs have bounded clique-width.}
\end{figure}

In Section~\ref{s-bounded} we prove each of the bounded cases in Theorem~\ref{thm:split-classification}. These proofs use results from the literature, which
we state in Section~\ref{sec:prelim}, together with some other preliminaries. In particular, we will 
exploit the close relationship between $H$-free split graphs and so-called weakly $H^\ell$-free bipartite graphs (see the next
section for a definition).
This enables us to apply Theorem~\ref{thm:bipweak} (a variant of Theorem~\ref{thm:bipartite}; both these theorems were proved in~\cite{DP14}) after first transforming a split graph into a bipartite graph by removing the edges of the clique
(this has to be done carefully, as a graph may have multiple split partitions).

In Section~\ref{s-classification} we prove Theorem~\ref{thm:split-classification}. We show that if the class of $H$-free split graphs
has bounded clique-width then~$H$ or~$\overline{H}$ must be an independent set
or an induced subgraph of~$F_4$ or~$F_5$. Both of these graphs have seven
vertices. The six-vertex induced subgraphs of~$F_4$ are: 
$\bull +\nobreak P_1, \overline{F_1}, \overline{F_3}$ and $K_{1,3}+\nobreak 2P_1$.
The six-vertex induced subgraphs
of~$F_5$ are: $\bull +\nobreak P_1, F_1,F_2,\overline{F_2},F_3,\overline{F_3}$ and $Q$.
These graphs and their complements are precisely the cases listed in Theorem~\ref{thm:split-classification} (and for which we prove boundedness in Section~\ref{s-bounded}).
Hence, we can also formulate our main theorem as follows. 

\medskip
\noindent
\faketheorem{Theorem~\ref{thm:split-classification} (alternative formulation).}
{\em
Let~$H$ be a graph such that neither~$H$ nor~$\overline{H}$ is in $\{F_4,F_5\}$.
The class of $H$-free split graphs has bounded clique-width if and only if
\\[-15pt]
\begin{itemize}
\item [$\bullet$] $H$ or~$\overline{H}$ is isomorphic to~$rP_1$ for some $r \geq 1$;
\item [$\bullet$] $H$ or $\overline{H} \ssi F_4$ or
\item [$\bullet$] $H$ or $\overline{H} \ssi F_5$.
\end{itemize}
}

\section{Preliminaries}\label{sec:prelim}

We only consider graphs that are finite, undirected and have neither multiple edges nor self-loops.
In this section we define some more graph terminology, additional notation and 
give some known lemmas from the literature that we will need to prove our results. We refer to the textbook of Diestel~\cite{Di12} for any undefined terminology.

Let $G=(V,E)$ be a graph.
The set $N(u)=\{v\in V\; |\; uv\in E\}$ is the {\em neighbourhood} of $u\in V$.
The {\em degree} of a vertex $u\in V$ in~$G$ is the size~$|N(u)|$ of its neighbourhood.
Let~$S,T\subseteq V$ with $S\cap T=\emptyset$. 
Then~$S$ is {\em complete} to~$T$ if every vertex in~$S$ is adjacent to every vertex in~$T$, and~$S$ is {\em anti-complete} to~$T$ if every vertex in~$S$ is non-adjacent to every vertex in~$T$.
Similarly, a vertex $v\in V\setminus T$ is {\em complete} or {\em anti-complete} to~$T$ if it is adjacent or non-adjacent, respectively, to every vertex of~$T$.
A set~$M$ of vertices is a {\em module} if every vertex not in~$M$ is either
complete or anti-complete to~$M$. A module of~$G$ is {\em trivial} if it
contains zero, one or all vertices of~$G$, otherwise it is {\em non-trivial}. 
A graph~$G$ is {\em prime} if every module in~$G$ is trivial.
We say that a vertex~$v$ {\em distinguishes} two vertices~$x$ and~$y$ if~$v$ is adjacent to precisely one of~$x$ and~$y$. Note that if a set~$M \subseteq V$ is not a module then there must be vertices $x,y \in M$ and a vertex $v \in V \setminus M$ 
such that~$v$ 
distinguishes~$x$ and~$y$.

In a partially ordered set $({\cal P},\leq)$, two elements $p,q \in {\cal P}$ are {\em comparable} if $p \leq q$ or $q \leq p$, otherwise they are {\em incomparable}.
A set $X \subseteq {\cal P}$ is a {\em chain} if the elements of~$X$ are pairwise comparable.

\subsection{Clique-Width}
The {\em clique-width} of a graph~$G$, denoted~$\cw(G)$, is the minimum
number of labels needed to
construct~$G$ by
using the following four operations:
\begin{enumerate}
\item creating a new graph consisting of a single vertex~$v$ with label~$i$;
\item taking the disjoint union of two labelled graphs~$G_1$ and~$G_2$;
\item joining each vertex with label~$i$ to each vertex with label~$j$ ($i\neq j$);
\item renaming label~$i$ to~$j$.
\end{enumerate}

A class of graphs~${\cal G}$ has \emph{bounded} clique-width if
there is a constant~$c$ such that the clique-width of every graph in~${\cal G}$ is at most~$c$; otherwise the clique-width of~${\cal G}$ is
{\em unbounded}.

Let~$G$ be a graph. We define the following operations.
For an induced subgraph $G'\ssi G$, the {\em subgraph complementation} operation (acting on~$G$ with respect to~$G'$) replaces every edge present in~$G'$
by a non-edge, and vice versa. Similarly, for two disjoint vertex subsets~$S$ and~$T$ in~$G$, the {\em bipartite complementation} operation with respect to~$S$ and~$T$ acts on~$G$ by replacing
every edge with one end-vertex in~$S$ and the other one in~$T$ by a non-edge and vice versa.

We now state some useful facts about how the above operations (and some other ones) influence the clique-width of a graph.
We will use these facts throughout the paper.
Let $k\geq 0$ be a constant and let~$\gamma$ be some graph operation.
We say that a graph class~${\cal G'}$ is {\em $(k,\gamma)$-obtained} from a graph class~${\cal G}$
if the following two conditions hold:
\begin{enumeratei}
\item every graph in~${\cal G'}$ is obtained from a graph in~${\cal G}$ by performing~$\gamma$ at most~$k$ times, and
\item for every $G\in {\cal G}$ there exists at least one graph
in~${\cal G'}$ obtained from~$G$ by performing~$\gamma$ at most~$k$ times.
\end{enumeratei}

We say that~$\gamma$ {\em preserves} boundedness of clique-width if
for any finite constant~$k$ and any graph class~${\cal G}$, any graph class~${\cal G}'$ that is $(k,\gamma)$-obtained from~${\cal G}$
has bounded clique-width if and only if~${\cal G}$ has bounded clique-width.

\newpage
\begin{enumerate}[\bf F{a}ct 1.]
\item \label{fact:del-vert} Vertex deletion preserves boundedness of clique-width~\cite{LR04}.\\[-1em]

\item \label{fact:comp} Subgraph complementation preserves boundedness of clique-width~\cite{KLM09}.\\[-1em]

\item \label{fact:bip} Bipartite complementation preserves boundedness of clique-width~\cite{KLM09}.\\[-1em]

\end{enumerate}

Combining the fact that the complement of any split graph is split with Fact~\ref{fact:comp} leads to the following lemma.

\begin{lemma}\label{lem:complement}
For any graph $H$, the class of $H$-free split graphs has bounded clique-width if and only if the class of $\overline{H}$-free
split graphs has bounded clique-width.
\end{lemma}

We will also need the following two results.

\begin{lemma}[\cite{CO00}]\label{lem:prime}
 If~${\cal P}$ is the set of all prime induced subgraphs of a graph~$G$ then $\cw(G)=\max_{H \in {\cal P}}\cw(H)$.
\end{lemma}

\begin{lemma}[\cite{MR99}]\label{lem:split-are-unbdd}
The class of split graphs has unbounded clique-width.
\end{lemma}

\subsection{Bipartite Graphs}
A graph is {\em bipartite} if its vertex set can be partitioned into two (possibly empty) independent sets.
Let~$H$ be a bipartite graph.
A \emph{black-and-white labelling}~$\ell$ of~$H$ is a labelling 
that assigns either the colour ``black'' or the colour ``white''  to each vertex of~$H$ in such a way that the
two resulting monochromatic colour classes~$B_H^\ell$ and~$W_H^\ell$ form a {\em bipartition} of~$V_H$ into two (possibly empty) independent sets. 
We say that~$H$ is a \emph{labelled} bipartite graph if we are also given a fixed black-and-white labelling.
We denote a graph~$H$ with such a labelling~$\ell$ by $H^\ell=(B_H^\ell,W_H^\ell,E_H)$. 
It is important to note that the pair $(B_H^\ell,W_H^\ell)$ is {\it ordered}, that is, $(B_H^\ell,W_H^\ell,E_H)$ and $(W_H^\ell,B_H^\ell,E_H)$ are different labelled 
bipartite graphs. Two labelled bipartite graphs~$H_1^{\ell_1}$ and~$H_2^{\ell_2}$ are \emph{isomorphic} if the following two conditions hold:
\begin{enumerate}[(i)]
\item the (unlabelled) graphs~$H_1$ and~$H_2$ are isomorphic, and
\item there exists an isomorphism $f: V_{H_1}\to V_{H_2}$ such that for all $u\in V_{H_1}$, it holds that
$u\in W^{\ell_1}_{H_1}$ if and only if $f(u)\in W^{\ell_2}_{H_2}$.
\end{enumerate}
Moreover, in this case~$\ell_1$ and~$\ell_2$ are said to be {\em isomorphic} labellings.
We write $H_1^{\ell_1} \li H_2^{\ell_2}$ if $H_1\ssi H_2$, $B_{H_1}^{\ell_1}\subseteq B_{H_2}^{\ell_2}$ and $W_{H_1}^{\ell_1}\subseteq W_{H_2}^{\ell_2}$.
In this case we say that~$H_1^{\ell_1}$ is a \emph{labelled} induced subgraph of~$H_2^{\ell_2}$. 
Note that the two labelled bipartite graphs~$H_1^{\ell_1}$ and~$H_2^{\ell_2}$ are isomorphic if and only if~$H_1^{\ell_1}$ is a labelled induced subgraph of~$H_2^{\ell_2}$, and vice versa. 

If~$H$ is a bipartite graph with a labelling~$\ell$, we let~$\overline{\ell}$ denote the ``opposite'' labelling labelling to~$\ell$, namely the labelling obtained from~$\ell$ by reversing the colours.
If~$H$ is a bipartite graph with the property that among all its black-and-white labellings, all those that maximize the number of black vertices are isomorphic, then we pick one such labelling and call 
it~$b$. 
If such a unique labelling~$b$ does exist, we let~$\overline{b}$ denote the opposite labelling to~$b$.

Let~$G$ be an (unlabelled) bipartite graph, and let~$H^\ell$ be a labelled bipartite graph. 
Then~$G$ is {\em weakly}
$H^\ell$-free if there is a labelling~$\ell^*$ of~$G$ such that~$G^{\ell^*}$
does not contain~$H^\ell$ as a labelled induced subgraph.
Similarly, let $\{H_1^{\ell_1},\ldots, H_p^{\ell_p}\}$ be a set of labelled bipartite graphs. Then~$G$ is {\it weakly $(H_1^{\ell_1},\ldots, H_p^{\ell_p})$-free} if there is a labelling~$\ell^*$ of~$G$ such that~$G^{\ell^*}$ does not contain any graph in $\{H_1^{\ell_1},\ldots, H_p^{\ell_p}\}$ as a labelled induced subgraph.

\medskip
\noindent
{\em Example.} The two non-isomorphic labelled bipartite graphs corresponding to~$P_1$ are shown in \figurename~\ref{fig:p1}.
Every edgeless graph is weakly $P_1^b$-free and weakly $P_1^{\overline{b}}$-free (simply label all the vertices white or all the vertices black, respectively).
However, if a bipartite graph is weakly $(P_1^b,P_1^{\overline{b}})$-free then it cannot contain any vertices.
Hence, a bipartite graph can be weakly $H_1^{\ell_1}$-free,$\ldots$, weakly $H_p^{\ell_p}$-free, while not being weakly $(H_1^{\ell_1},\ldots, H_p^{\ell_p})$-free.

\begin{figure}
\begin{center}
\begin{tabular}{cc}
\begin{minipage}{0.3\textwidth}
\centering
\begin{tikzpicture}[scale=0.4]
\GraphInit[vstyle=Simple]
\SetVertexSimple[MinSize=6pt]
\Vertex[x=0,y=0]{x}
\end{tikzpicture}
\end{minipage}
&
\begin{minipage}{0.3\textwidth}
\centering
\begin{tikzpicture}[scale=0.4]
\GraphInit[vstyle=Simple]
\SetVertexSimple[MinSize=6pt,FillColor=white]
\Vertex[x=0,y=0]{x}
\end{tikzpicture}
\end{minipage}\\
\\
$P_1^b$ & $P_1^{\overline{b}}$
\end{tabular}
\caption{The two pairwise non-isomorphic labellings of~$P_1$.}\label{fig:p1}
\end{center}
\end{figure}
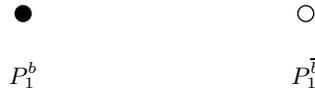

\medskip
\noindent
For a more in-depth discussion of weakly $H^\ell$-free bipartite graphs we refer to~\cite{DP14}. 
In this paper we will make use of the following theorem (see also \figurename~\ref{fig:bip-weakly}).

\begin{theorem}[\cite{DP14}]\label{thm:bipweak}
Let $H^\ell$ be a labelled bipartite graph.
The class of weakly $H^\ell$-free bipartite graphs has bounded clique-width if and only if one of the following cases holds:
\begin{itemize}
\item [$\bullet$] $H^\ell$ or $H^{\overline{\ell}} = (sP_1)^b$ for some $s\geq 1$;
\item [$\bullet$] $H^\ell$ or $H^{\overline{\ell}} \li (P_1+P_5)^b$;
\item [$\bullet$]  $H^\ell \hspace*{1mm} \li (P_2+P_4)^b$ or
\item [$\bullet$] $H^\ell \hspace*{1mm} \li (P_6)^b$.
\end{itemize}
\end{theorem}
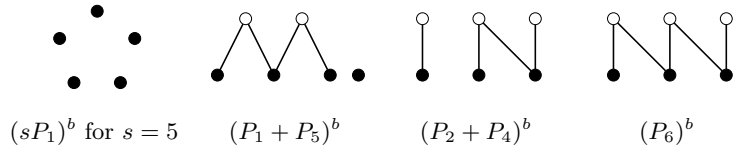
\begin{figure}[h]
\begin{center}
\begin{tabular}{cccc}
\begin{minipage}{0.20\textwidth}
\centering
\scalebox{0.75}{
{\begin{tikzpicture}[scale=1,rotate=90]
\GraphInit[vstyle=Simple]
\SetVertexSimple[MinSize=6pt]
\Vertex[a=0,d=0.7]{a}
\Vertex[a=72,d=0.7]{b}
\Vertex[a=144,d=0.7]{c}
\Vertex[a=216,d=0.7]{d}
\Vertex[a=288,d=0.7]{e}
\end{tikzpicture}}}
\end{minipage}
&
\begin{minipage}{0.20\textwidth}
\centering
\scalebox{0.75}{
\begin{tikzpicture}[scale=1]
\GraphInit[vstyle=Simple]
\SetVertexSimple[MinSize=6pt]
\Vertex[x=0,y=0]{x0}
\Vertex[x=1,y=0]{x2}
\Vertex[x=2,y=0]{x4}
\Vertex[x=2.5,y=0]{x5}
\SetVertexSimple[MinSize=6pt,FillColor=white]
\Vertex[x=0.5,y=1]{x1}
\Vertex[x=1.5,y=1]{x3}
\Edges(x0,x1,x2,x3,x4)
\end{tikzpicture}}
\end{minipage}
&
\begin{minipage}{0.20\textwidth}
\centering
\scalebox{0.75}{
\begin{tikzpicture}[scale=1]
\GraphInit[vstyle=Simple]
\SetVertexSimple[MinSize=6pt]
\Vertex[x=0,y=0]{x0}
\Vertex[x=1,y=0]{x2}
\Vertex[x=2,y=0]{x4}
\SetVertexSimple[MinSize=6pt,FillColor=white]
\Vertex[x=0,y=1]{x1}
\Vertex[x=1,y=1]{x3}
\Vertex[x=2,y=1]{x5}
\Edge(x0)(x1)
\Edges(x2,x3,x4,x5)
\end{tikzpicture}}
\end{minipage}
&
\begin{minipage}{0.20\textwidth}
\centering
\scalebox{0.75}{
\begin{tikzpicture}[scale=1]
\GraphInit[vstyle=Simple]
\SetVertexSimple[MinSize=6pt]
\Vertex[x=0,y=0]{x0}
\Vertex[x=1,y=0]{x2}
\Vertex[x=2,y=0]{x4}
\SetVertexSimple[MinSize=6pt,FillColor=white]
\Vertex[x=0,y=1]{x1}
\Vertex[x=1,y=1]{x3}
\Vertex[x=2,y=1]{x5}
\Edges(x0,x1,x2,x3,x4,x5)
\end{tikzpicture}}
\end{minipage}\\
& & &\\
$(sP_1)^b$ for $s=5$ & $(P_1+P_5)^b$ & $(P_2+P_4)^b$ & $(P_6)^b$
\end{tabular}
\caption{The labelled bipartite graphs from Theorem~\ref{thm:bipweak}.}
\label{fig:bip-weakly}
\end{center}
\end{figure}
Similarly to the way that a bipartite graph can have multiple labellings, a split
graph~$G$ may have multiple split partitions, say $(K_1,I_1)$ and $(K_2,I_2)$.
We say that two such split partitions are {\em isomorphic} if there is an
isomorphism $f:\nobreak V(G) \to V(G)$ of~$G$ such that $u \in K_1$ if and only if $f(u) \in K_2$.
Let~$G$ and~$H$ be split graphs with split partitions $(K_G,I_G)$ and $(K_H,I_H)$, respectively.
Then
$(K_G,I_G)$ {\it contains} $(K_H,I_H)$ if $H \ssi G$, $K_H\subseteq K_G$ and $I_H\subseteq I_G$.
We will explore the properties of split partitions in the proof of Lemma~\ref{lem:split-f4-f5-lemma}.

\section{Proofs of the Bounded Cases in Theorem~\ref{thm:split-classification}}\label{s-bounded}

In this section we show that the clique-width of each of the seven classes of $H$-free graphs given in Theorem~\ref{thm:split-classification} is bounded.
We start with the case $H=rP_1$, for which we give an explicit bound.\footnote{For the other bounded cases we do not specify any upper bounds. This would complicate our proofs for negligible gain, as our primary goal is to show boundedness. Moreover, in our proofs we apply graph operations that may exponentially increase the upper bound on the clique-width, which means that any bounds obtained from our proofs would be very large and far from being tight.
Furthermore, we make use of other results that do not give explicit bounds.}

\begin{theorem}\label{thm:rP_1}
For any $r \geq 1$, the class of $rP_1$-free graphs has clique-width at most $r+1$.
\end{theorem}

\begin{proof}
Let $H=rP_1$ for some $r \geq 1$ and let~$G$ be an $H$-free split graph with split partition $(K,I)$.
It follows that $|I| < r$.
In this case it is easy to see that the clique-width of~$G$ is at most~$r+1$:
We introduce the (at most $r-1$) vertices of~$I$ with distinct labels.
We use one more label for ``new'' vertices of~$K$ and one more label for ``processed'' vertices of~$K$.
We then add each vertex of~$K$ one-by-one, labelling it with the ``new'' label, and immediately connect it to all the already ``processed'' vertices of~$K$, along with any relevant vertices of~$I$, after which we relabel the new vertex to be ``processed.''\qed
\end{proof}

We now consider the cases $H=\bull+\nobreak P_1$ and $H=Q$. In order to prove these two cases we apply Theorem~\ref{thm:bipweak} for the first time.

\begin{theorem}\label{thm:split-bounded-from-weakly-bip}
The class of 
$(\bull+\nobreak P_1)$-free split graphs and the class of~$Q$-free split graphs have bounded clique-width.
\end{theorem}
\begin{proof}
Let~$H$ be $\bull+\nobreak P_1$ or~$Q$ and
let~$H_0^\ell$ be the labelled bipartite graph
$(P_1+\nobreak P_5)^b$ or $(P_2+\nobreak P_4)^b$, respectively.
Suppose~$G$ is an $H$-free split graph and fix a split partition
$(K,I)$ of~$V(G)$. Let~$G'$ be the graph
obtained from~$G$ by applying a complementation to~$G[K]$. By
Fact~\ref{fact:comp}, we need only show that~$G'$ has bounded clique-width.
Now~$G'$ is a bipartite graph with bipartition~$(K,I)$. If we label the
vertices of~$K$ white and the vertices of~$I$ black, then we find that~$G'$ is a weakly
$H_0^\ell$-free bipartite graph and therefore has bounded clique-width by
Theorem~\ref{thm:bipweak}.
\qed
\end{proof}

The next theorem follows from Theorem~\ref{thm:chordal-classification} and Lemma~\ref{lem:complement} (recall that every split graph is chordal).
However, the proof of the corresponding case for chordal graphs is much more complicated.
In light of this, and to make this paper more self-contained, we include a (much simpler) direct proof for this case.

\begin{theorem}\label{thm:k13+2p1-split}
The class of $(K_{1,3}+2P_1)$-free split graphs has bounded clique-width.
\end{theorem}

\begin{proof}
Let~$G$ be a $(K_{1,3}+2P_1)$-free split graph and fix of a partition of its
vertices into a clique~$K$ and an independent set~$I$. 
If $|I| \leq 5$ then~$G$ is $7P_1$-free (at most one vertex of any independent set in~$G$ can belong to~$K$), in which case we are done by Theorem~\ref{thm:rP_1}.
We therefore assume that $|I| \geq 6$.
Since~$G$ is
$(K_{1,3}+2P_1)$-free, every vertex in~$K$ has either at most two neighbours
in~$I$ or at most one non-neighbour in~$I$. Let~$K'$ be the set of vertices
in~$K$ that have exactly two neighbours in~$I$.
Suppose $x,y \in K'$ and let~$w$ and~$w'$ be the two neighbours of~$x$ in~$I$ and let~$z$ and~$z'$ be two common non-neighbours of~$x$ and~$y$ in~$I$ (which exist since $|I| \geq 6$).
Then one of~$y$'s neighbours in~$I$ must be~$w$ or~$w'$ otherwise $G[x,y,w,w',z,z']$ would be a $K_{1,3}+\nobreak 2P_1$, a contradiction.

If~$K'$ is is non-empty, choose $x \in K'$ arbitrarily and delete
both neighbours of~$x$ in~$I$ (we may do this by Fact~\ref{fact:del-vert}) to
obtain a graph~$G'$. Now every vertex of~$K'$ has at most one neighbour in $I'=I \cap V(G')$
in the graph~$G'$. (If~$K'$ was already empty, then we set $G'=G, I'=I$.) In the
graph~$G'$ every vertex in~$K$ has either at most one neighbour or at most one
non-neighbour in~$I'$. Let~$K''$ be the set of vertices that have more than one
neighbour in~$I'$. By Fact~\ref{fact:bip}, we may apply a bipartite
complementation between~$K''$ and~$I'$ to obtain a graph~$G''$ in which every
vertex of~$K$ has at most one neighbour in~$I'$. Finally apply a complementation
to the set~$K$ (we may do this by Fact~\ref{fact:comp}). The resulting graph is
a disjoint union of stars, so it has clique-width at most~$2$. This completes
the proof.\qed
\end{proof}

It remains to prove that the class of $F_i$-free graphs has bounded clique-width for $i\in{1,2,3}$. We do this in Theorems~\ref{thm:f1-split}--\ref{thm:f3-split}.

\begin{theorem}\label{thm:f1-split}
The class of $F_1$-free split graphs has bounded clique-width.
\end{theorem}

\begin{proof}
Let~$G$ be an $F_1$-free split graph. Fix a split partition $(K,I)$ of~$G$. By
Lemma~\ref{lem:prime}, we may assume that~$G$ is prime. If~$G$ contains an
induced bull (see also \figurename~\ref{fig:bull-dart})
that has three vertices in~$K$ and two in~$I$, we say that this
bull is {\em special}.

First suppose that~$G$ does not contain~$18$ vertex-disjoint special bulls.
By Fact~\ref{fact:del-vert}, we may delete at most $5 \times 17 = 85$ vertices from~$G$ to obtain a split graph with no special bulls.
Since the resulting graph contains no special bulls, it must be $Q$-free, and therefore has bounded clique-width by Theorem~\ref{thm:split-bounded-from-weakly-bip}.

We may therefore assume that~$G$ contains~$18$ vertex-disjoint special bulls,
$B_1,\ldots,B_{18}$, say. For $h \in \{1,\ldots,18\}$, let $J_h=
\{j_{1,h},j_{2,h},j_{3,h}\} = K \cap V(B_h)$ and $I_h=\{i_{1,h},i_{2,h}\}=I
\cap V(B_h)$.
In the remainder of the proof, we will show that~$G$ must contain a non-trivial module, contradicting the fact that~$G$ is prime.

We first state the following two observations, both of which follow directly from the fact that~$G$ is an $F_1$-free split graph.

\obs{\label{obs:F1-nested-nbhds} If $s, t \in I$ have two common non-neighbours in~$K$ then $N(s)\subseteq N(t)$ or $N(t)\subseteq N(s)$.}

\obs{\label{obs:F1-non-neighbour-in-bull} Every $x \in I$ has a non-neighbour in every~$J_h$.}

\noindent Consider the special bulls~$B_1$ and~$B_2$. By
Observation~\ref{obs:F1-non-neighbour-in-bull}, every vertex in~$I$ must have a
non-neighbour in~$J_1$ and a non-neighbour in~$J_2$. Let~$I_{i,j}$ denote the
set of vertices in~$I$ that are non-adjacent to both~$j_{i,1}$ and~$j_{j,2}$,
for $i,j \in \{1,2,3\}$. (Note that every vertex of~$I$ must be in at least one
set~$I_{i,j}$, but it may be in more than one such set.) By
Observation~\ref{obs:F1-nested-nbhds}, for any two vertices $s,t$ in any
set~$I_{i,j}$ either $N(s)\subseteq N(t)$ or $N(t)\subseteq N(s)$.

Since~$G$ is prime, no two vertices of~$I$ have the same neighbourhood.  We may
therefore define a partial order~$\leq_N$ on~$I$: given two vertices $s,t\in I$,
we say that $s\leq_N t$ if $N(s)\subseteq N(t)$. Note every set~$I_{i,j}$ is a chain under
this partial order, so~$I$ can be covered by at most nine chains.

We rename the sets~$I_{i,j}$ to be $S_1,\ldots,S_p$, in an arbitrary order, deleting any sets~$I_{i,j}$ that are empty, so $p \leq 9$.
For $i \in \{1,\ldots,p\}$, let~$s_i$ be the maximum element of~$S_i$ (under
the~$\leq_N$ ordering).  From the definition of the sets $I_{i,j}$ it follows
that for $k \in \{1,\ldots,p\}$, $S_k = \{x \in I \; | \; N(x) \subseteq N(s_k)\}$. If
there are distinct $i,j$ such that $N(s_i) \subseteq N(s_j)$ then~$S_i
\subseteq S_j$, so we may delete the set~$S_i$ from the set of chains~$S_k$
that we consider and every vertex of~$I$ will still be in some set~$S_k$. In
other words, we may assume that $S_1,\ldots,S_q$ are chains under the~$\leq_N$
ordering, with maximal elements $s_1,\ldots,s_q$, respectively, where $q \leq p \leq 
9$ and every pair $s_i,s_j$ is incomparable under the~$\leq_N$ ordering. (Note
that $q \geq 2$, since $i_{1,1},i_{2,1} \in V(B_1)$ have incomparable
neighbourhoods.) 

By Observation~\ref{obs:F1-non-neighbour-in-bull}, for each~$i
\in \{1,\ldots,q\}$, the vertex~$s_i$ must be non-adjacent to at least one
vertex in each of $J_1,\ldots,J_{18}$, so it must have at least~$18$
non-neighbours in~$K$. Let~$X_i$ be the set of vertices in~$K$ that are
non-adjacent to~$s_i$ and note that since~$s_i$ is maximum in~$S_i$, the
set~$X_i$ is anti-complete to~$S_i$.

Since for $i \in \{2,\ldots,q\}$ the vertices~$s_i$ and~$s_1$ are incomparable,
it follows that~$s_i$ is adjacent to all but at most one vertex of~$X_1$ (by
Observation~\ref{obs:F1-nested-nbhds}).  Therefore, there is a subset $X'_1
\subseteq X_1$ with $|X'_1|\geq 18-8=10$ such that~$X'_1$ is complete to
$\{s_2,\ldots,s_q\}$, so $X'_1 \subseteq N(s_i) \setminus N(s_1)$ for $i
\in \{2,\ldots,q\}$.

\clm{\label{clm:1-f1-split}
For $i \in \{2,\ldots,q\}$ there is a vertex $z_i \in K$ such
that every vertex in~$S_i$ is either complete or anti-complete to $N(s_i) \setminus
(N(s_1) \cup \{z_i\})$.}

\noindent
We prove Claim~\ref{clm:1-f1-split} as follows.
Let~$t$ be the smallest (with respect to~$\leq_N$) vertex in~$S_i$ that has a
neighbour, say~$w$, in $N(s_i) \setminus N(s_1)$. Any vertex $s \in S_i$ with
$s <_N t$ is anti-complete to $N(s_i) \setminus N(s_1)$.  Now $t \not \leq_N
s_1$, since~$w$ is not a neighbour of~$s_1$ and $s_1 \not \leq_N t$, since $t
\leq_N s_i$ and $s_1 \not \leq_N s_i$. By
Observation~\ref{obs:F1-nested-nbhds}, $N(t) \cup N(s_1)$ contains all but at
most one vertex of~$K$. If there is a vertex $z_i \in K \setminus (N(t) \cup
N(s_1))$ then~$t$ is complete to $N(s_i) \setminus (N(s_1) \cup \{z_i\})$
(if there is no such vertex then we
choose $z_i \in K$ arbitrarily, and the same conclusion holds). If $s \in S_i$
and $t \leq_N s$ then $N(s) \supseteq N(t) \supseteq N(s_i) \setminus (N(s_1)
\cup \{z_i\})$, as desired. This completes the proof of 
Claim~\ref{clm:1-f1-split}.

\medskip
Recall that for $i \in \{2,\ldots,q\}$, $X'_1 \subseteq N(s_i) \setminus
N(s_1)$. Let $X''_1 = X'_1 \setminus \{z_2,\ldots,z_q\}$ (where $z_2,\ldots,z_q$ are
defined as in Claim~\ref{clm:1-f1-split} above). Then $|X''_1| \geq 10 - 8 = 2$ and every vertex
in~$I$ is either complete or anti-complete to~$X''_1$. Therefore~$X''_1$ is a
non-trivial module of~$G$, contradicting the fact that~$G$ is prime. This
completes the proof.
\qed
\end{proof}

\begin{theorem}\label{thm:f2-split}
The class of $F_2$-free split graphs has bounded clique-width.
\end{theorem}

\begin{proof}
Let~$G$ be an $F_2$-free split graph. Fix a split partition $(K,I)$ of~$G$. By
Lemma~\ref{lem:prime}, we may assume that~$G$ is prime. If~$G$ contains an
induced~$Q$ (see also \figurename~\ref{fig:bounded-split}) it must have three vertices in~$K$ and three in~$I$
(since~$Q$ has a unique split partition).

First suppose that~$G$ does not contain two vertex-disjoint copies of~$Q$.
By Fact~\ref{fact:del-vert}, we may delete at most six vertices from~$G$ to obtain a $Q$-free split graph.
By Theorem~\ref{thm:split-bounded-from-weakly-bip}, 
the resulting graph (and thus~$G$) has bounded clique-width.

We may therefore assume that~$G$ contains two vertex-disjoint copies of~$Q$,
say~$Q_1$ and~$Q_{2}$. For 
$h \in \{1,2\}$, let $J_h=
\{j_{1,h},j_{2,h},j_{3,h}\} = K \cap V(Q_h)$ and $I_h=\{i_{1,h},i_{2,h},i_{3,h}\}=I
\cap V(Q_h)$, where 
$E(Q_h)=\{i_{1,h}j_{1,h},i_{2,h}j_{2,h},i_{3,h}j_{2,h},i_{3,h}j_{3,h}\}\cup \{j_{1,h}j_{2,h},j_{1,h}j_{3,h},j_{2,h}j_{3,h}\}$.

We say that two vertices $s,t \in I$ have {\em comparable} neighbourhoods if $N(s) \subseteq N(t)$ or $N(t) \subseteq N(s)$. Otherwise we say that~$s$ and~$t$ have {\em incomparable} neighbourhoods.

\clm{\label{clm:incomp} Suppose $s,t \in I$ have a common non-neighbour $u \in K$. If~$s$ and~$t$ have incomparable neighbourhoods then $|N(s) \setminus N(t)| = |N(t) \setminus N(s)|=1$.}

\noindent
We proof Claim~\ref{clm:incomp} as follows.
Since~$s$ and~$t$ have incomparable neighbourhoods, there must be a vertex $v \in N(s) \setminus N(t)$ and a vertex $w \in N(t) \setminus N(s)$.
Suppose, for contradiction, that there is another vertex $w' \in N(t) \setminus N(s)$. 
Then $G[s,t,u,v,w,w']$ is an~$F_2$. This contradiction completes the proof of 
Claim~\ref{clm:incomp}.

\medskip
\noindent
The vertices~$i_{1,1}$ and~$i_{3,1}$ cannot have a common non-neighbour $t \in K$, otherwise $G[i_{1,1},i_{3,1},j_{1,1},j_{2,1},j_{3,1},t]$ would be an~$F_2$.
It follows that:
\begin{equation}\label{eqn:K-is-covered}
N(i_{1,1}) \cup N(i_{3,1}) = K.
\end{equation}

\noindent
Next, by Claim~\ref{clm:incomp}, since~$i_{1,1}$ and~$i_{2,1}$ have incomparable neighbourhoods and a common non-neighbour in~$K$, namely~$j_{3,1}$ it follows that:
\begin{equation}\label{eqn:similar-neighbourhoods}
N(i_{1,1}) = (N(i_{2,1}) \setminus \{j_{2,1}\}) \cup \{j_{1,1}\}.
\end{equation}
Combining~(\ref{eqn:K-is-covered}) and~(\ref{eqn:similar-neighbourhoods}), we conclude that:
\begin{equation}\label{eqn:K-is-almost-covered}
K \setminus \{j_{1,1}\} \subseteq N(i_{2,1}) \cup N(i_{3,1}).
\end{equation}
\noindent
Now~$i_{2,1}$ and~$i_{3,1}$ have a common non-neighbour, namely $j_{1,1}$. Note that $j_{3,1} \in N(i_{3,1}) \setminus N(i_{2,1})$.
By Claim~\ref{clm:incomp} it follows that either $N(i_{2,1}) \subseteq N(i_{3,1})$ (if~$i_{2,1}$ and~$i_{3,1}$ have comparable neighbourhoods) or $N(i_{3,1})\setminus \{j_{3,1}\} \subset N(i_{2,1})$ (if they do not).
This means that $K \setminus \{j_{1,1},j_{3,1}\}$ is a subset of $N(i_{2,1})$ or $N(i_{3,1})$.
In particular, $i_{2,1}$ or~$i_{3,1}$, respectively, is complete to~$J_2 \subset K$.
Then this vertex, together with~$J_2$, $i_{1,2}$ and~$i_{3,2}$ induces an~$F_2$ in~$G$. This contradiction completes the proof.\qed
\end{proof}

\begin{theorem}\label{thm:f3-split}
The class of $F_3$-free split graphs has bounded clique-width.
\end{theorem}

\begin{proof}
Let~$G$ be an $F_3$-free split graph. Fix a split partition $(K,I)$ of~$G$. By
Lemma~\ref{lem:prime}, we may assume that~$G$ is prime. If~$G$ contains an
induced dart (see also \figurename~\ref{fig:bull-dart}) which has has three vertices in~$K$ and two in~$I$, we say that this dart is {\em special}.

First suppose that~$G$ does not contain~$19$ vertex-disjoint special darts.
By Fact~\ref{fact:del-vert}, we may delete at most $5 \times 18 = 90$ vertices from~$G$ to obtain a split graph with no special dart.
Since the resulting graph contains no special copies of the dart, it must be $Q$-free, and therefore has bounded clique-width by Theorem~\ref{thm:split-bounded-from-weakly-bip}.

We may therefore assume that~$G$ contains~$19$ vertex-disjoint special darts,
$D_1,\ldots,D_{19}$, say. For $h \in \{1,\ldots,19\}$, let $J_h=
\{j_{1,h},j_{2,h},j_{3,h}\} = K \cap V(D_h)$ and $I_h=\{i_{1,h},i_{2,h}\}=I
\cap V(D_h)$.
We will use the following claim.

\clm{\label{clm:one-nbhr-one-non-nbr} If $i,j \in \{1,\ldots,19\}$ then every vertex of~$I_i$ has at least one neighbour and at least one non-neighbour in~$J_j$.}

\noindent 
We prove Claim~\ref{clm:one-nbhr-one-non-nbr} as follows.
If $i=j$ then the claim follows from the definition of~$D_i$.
Suppose $i \neq j$.
If a vertex $x \in I_i$ is complete to~$J_j$ then $G[\{x\} \cup J_j \cup I_j]$ is an~$F_3$, which is a contradiction.
Therefore each vertex in~$I_i$ has at least one non-neighbour in~$J_j$.
Now suppose for contradiction that a vertex $x \in I_i$ has no neighbours in~$J_j$.
Let~$x'$ be the other vertex of~$I_i$.
It must have a non-neighbour $y \in J_j$.
Note that~$y$ is then anti-complete to~$I_i$.
Now $G[\{y\} \cup J_i \cup I_i]$ is an~$F_3$.
This contradiction completes the proof of 
Claim~\ref{clm:one-nbhr-one-non-nbr}.

\medskip
Claim~\ref{clm:one-nbhr-one-non-nbr} implies that for every~$i,j \in \{1,\ldots,19\}$, every vertex of~$I_i$ must have one of the six possible neighbourhoods in~$J_j$, namely those that contain at least one vertex of~$J_j$, but not all vertices of~$J_j$.
This means we can partition the vertices of $I_1\cup\cdots\cup I_{19}$ into 36 sets (some of which may be empty), according to their neighbourhood in $J_1 \cup J_2$.
Since $I_1\cup\cdots\cup I_{19}$ consists of 38 vertices, two of these vertices, say~$x$ and~$x'$ must have the same neighbourhood in $J_1 \cup J_2$.
Furthermore, by Claim~\ref{clm:one-nbhr-one-non-nbr}, they have a common neighbour~$y \in J_1$ and common  non-neighbours $z\in J_1$ and $z' \in J_2$.
Since the graph~$G$ is prime, the set $\{x,x'\}$ cannot be a module.
Therefore there must be a vertex~$z''$ that distinguishes~$x$ and~$x'$, say~$z''$ is adjacent to~$x$, but non-adjacent to~$x'$. Note that $z'' \in K$, so it must be adjacent to $y,z$ and~$z'$.
Now $G[x,x',y,z,z',z'']$ is an~$F_3$.
This contradiction completes the proof.\qed
\end{proof}

\section{Completing the Proof of Theorem~\ref{thm:split-classification}}\label{s-classification}

In this section we use the results from the previous section to prove our main result. We also need the following lemma.

\begin{lemma}[Key Lemma]\label{lem:split-f4-f5-lemma}
If the class of~$H$-free split graphs has bounded clique-width then~$H$
or~$\overline{H}$ is isomorphic to~$K_r$ for some~$r$ or is an induced subgraph
of~$F_4$ or~$F_5$.
\end{lemma}

\begin{proof}
Suppose that~$H$ is a graph such that the class of~$H$-free split graphs has
bounded clique-width. Then~$H$ must be a split graph, otherwise the class of
$H$-free split graphs would include all split graphs, in which case the
clique-width would be unbounded by Lemma~\ref{lem:split-are-unbdd}.

Suppose that~$H$ has two split partitions $(K,I)$ and $(K',I')$ that are not
isomorphic. There cannot be two distinct vertices $x,y \in I \setminus I'$, as
then $x,y \in I$, so they would have to be non-adjacent, and similarly $x,y \in K'$, so
they would have to be adjacent, a contradiction.
Hence, $|I\setminus I'|\leq 1$.
For the same reason, $|I'\setminus I|\leq 1$.

Next suppose that $|I\setminus I'|=|I'\setminus I|=1$. Then there exist
vertices $x\in I \setminus I'$ and $y \in I' \setminus I$. Let $I'' = I \setminus \{x\}$
and $K''=K \setminus \{y\}$. Then $I=I'' \cup \{x\}, K=K'' \cup \{y\}, I'=I''
\cup \{y\}$ and $K'=K'' \cup \{x\}$. Since~$x\in I$ and $x \in K'$, $x$ must be
anti-complete to~$I''$ and complete to~$K''$. Since $y\in I'$ and $y \in K$ the
same is true for $y$. ($x$ and~$y$ may or may not be adjacent to each-other.)
However, this means that $(K,I)$ and $(K',I')$ are isomorphic split partitions
of~$H$, which is a contradiction. 

Due to the above, we may assume without loss of generality that $|I\setminus I'|=1$ and $|I'\setminus I|=0$.
Hence there is a vertex~$x$ such that $I=I' \cup \{x\}$ and $K' = K \cup
\{x\}$. Let $H'=H \setminus \{x\}$ and note that~$H'$ has split partition
$(K,I')$ (though~$H'$ may also have a different split partition) and that~$H$ can be obtained
from~$H'$ by adding a vertex that is adjacent to every vertex of~$K$ and
non-adjacent to every vertex of~$I'$.

Let~$H'^\ell_b$ be the labelled bipartite graph obtained from~$H'$ by
complementing~$K$, colouring every vertex of~$I'$ white and every vertex of~$K$
black. 
Let~$G$ be a weakly $H'^\ell_b$-free graph. Then 
$G$ has a black-and-white labelling $\ell^*$ such that 
$G^{\ell^*}=(B_G^{\ell^*},W_G^{\ell^*},E_G)$ does not contain~$H'^\ell_b$
as a labelled induced subgraph.
Let~$G_S$ be the split graph obtained from~$G^{\ell^*}$ by
complementing the set of black vertices. Then
$(B_G^{\ell^*},W_G^{\ell^*})$ is a split partition of~$G_S$ 
that does not contain $(K,I')$. 
Therefore, $G_S$ has a split partition that does not contain $(K,I)$ or $(K',I')$.
Hence, $G_S$ is $H$-free. 
Since we assumed that the class of $H$-free split graphs has bounded
clique-width, by Fact~\ref{fact:comp} it
follows that the class of weakly $H'^\ell_b$-free bipartite graphs must have
bounded clique-width. By Theorem~\ref{thm:bipweak}, $H'^\ell_b$  must therefore
be a black independent set, a white independent set or a labelled induced
subgraph of $(P_1+\nobreak P_5)^b,\allowbreak (P_1+P_5)^{\overline{b}},
(P_2+\nobreak P_4)^b$ or~$P_6^b$. This corresponds to the cases where~$H$ is a
clique, an independent set or an induced subgraph of $\overline{F_4}, F_4,
F_5$ or~$\overline{F_5}$, respectively.

Now suppose that all split partitions of~$H$ are isomorphic. We argue similarly. Let $(K,I)$ be a split partition of~$H$.
Let~$H^\ell_b$
be the labelled bipartite graph obtained from~$H$ by complementing~$K$,
colouring every vertex of~$I$ white and every vertex of~$K$ black.
Let~$G^{\ell^*}=(B_G^{\ell^*},W_G^{\ell^*},E_G)$ be a labelled bipartite graph and
let~$G_S$ be the split graph obtained from~$G^{\ell^*}$ by complementing the set
of black vertices. Then~$G^{\ell^*}$ 
does not contain~$H^\ell_b$ as a labelled induced subgraph
if and only if~$G_S$ is $H$-free. Proceeding as before, we find
that again~$H$ must be an independent set or a clique or an induced subgraph of
either $F_4,\overline{F_4},F_5$ or~$\overline{F_5}$ (though this time the extra
vertex~$x$ is not present, so~$H$ must be a proper induced subgraph of one of these graphs).
\qed\end{proof}

We are now ready to prove Theorem~\ref{thm:split-classification}.

\medskip
\noindent
\faketheorem{Theorem \ref{thm:split-classification} (restated).}
{\em Let~$H$ be a graph such that neither~$H$ nor~$\overline{H}$ is in $\{F_4,F_5\}$.
The class of $H$-free split graphs has bounded clique-width if and only if
\begin{itemize}
\item [$\bullet$] $H$ or~$\overline{H}$ is isomorphic to~$rP_1$ for some $r \geq 1$;
\item [$\bullet$] $H$ or $\overline{H} \ssi \bull+P_1$;
\item [$\bullet$] $H$ or $\overline{H} \ssi F_1$; 
\item [$\bullet$] $H$ or $\overline{H} \ssi F_2$; 
\item [$\bullet$] $H$ or $\overline{H} \ssi F_3$; 
\item [$\bullet$] $H$ or $\overline{H} \ssi Q$ or
\item [$\bullet$] $H$ or $\overline{H} \ssi K_{1,3}+2P_1$.
\end{itemize}
}

\begin{proof}
If $H=rP_1$, the result follows from Theorem~\ref{thm:rP_1}.
The corresponding result for~$\overline{H}$ then follows from Lemma~\ref{lem:complement}.
For the remainder of the proof, we may therefore assume that~$H$ contains at least one edge and at least one non-edge.

If~$H$ is an induced subgraph of $\bull+P_1, F_1, F_2, F_3, Q$ or $K_{1,3}+\nobreak 2P_1$ then the result follows from Theorems \ref{thm:split-bounded-from-weakly-bip}, \ref{thm:f1-split}, \ref{thm:f2-split}, \ref{thm:f3-split}, \ref{thm:split-bounded-from-weakly-bip} and~\ref{thm:k13+2p1-split}, respectively.
The corresponding results for~$\overline{H}$ then follow from Lemma~\ref{lem:complement}.

Now suppose that the class of $H$-free split graphs has bounded clique-width.
Recall that~$H$ contains at least one edge and at least one non-edge.
By Lemma~\ref{lem:split-f4-f5-lemma} combined with Lemma~\ref{lem:complement}, we may assume that~$H$ is an induced subgraph of~$F_4$ or~$F_5$.
Note that both~$F_4$ and~$F_5$ have seven vertices.
The six-vertex induced subgraphs of~$F_4$ are: $\bull +\nobreak P_1,
\overline{F_1}, \overline{F_3}$ and $K_{1,3}+\nobreak 2P_1$.
The six-vertex induced subgraphs
of~$F_5$ are: $\bull+\nobreak P_1,F_1,F_2,\overline{F_2},F_3,\overline{F_3}$ and~$Q$.
These graphs and their complements are precisely the bounded cases considered above.
\qed
\end{proof}

\section{Future Work}\label{s-con}

Our goal is to solve the two open cases for $H$-free split graphs, namely $H\in \{F_4,F_5\}$ and the two open cases for $H$-free chordal graphs, namely $H\in \{F_1,F_2\}$. This does not seem a straightforward task. There is still some hope that, for $i\in\{1,2\}$, the class of $F_i$-free chordal graphs has bounded clique-width 
as we proved that the class of $F_i$-free split graphs has bounded clique-width. Note that for $i\in\{4,5\}$, the class of $F_i$-free chordal graphs has unbounded clique-width and it does not seem possible to modify the construction that shows this to get a proof for split graphs (hence we could potentially have two other subclasses of split graphs with bounded clique-width). 
We also recall that such results, just as in other cases~\cite{BDHP15,DHP0,DP14}, could be useful for completing the classification for $(H_1,H_2)$-free graphs.

\bibliography{mybib-split}
\end{document}